\documentclass[%
reprint,
superscriptaddress,
bibnotes,
amsmath,amssymb,
pra,
floatfix,
]{revtex4-1}
\usepackage{graphicx}
\usepackage{dcolumn}
\usepackage{bm}
\usepackage[breaklinks=true]{hyperref}
\usepackage{natbib}
\usepackage{mathrsfs,mathtools,color,wasysym,dsfont,here}
\usepackage{amsthm}
\usepackage{physics}
\usepackage[all]{xy}
\usepackage{amsfonts}
\usepackage{array,booktabs}
\usepackage{comment}
\usepackage{cases}

\preprint{}

\newcommand{\sites}{\Lambda}
\newcommand{\bonds}{B}

\theoremstyle{plain}
\newtheorem{thm}{Theorem}
\newtheorem{lem}[thm]{Lemma}
\newtheorem{cor}[thm]{Corollary}

\theoremstyle{definition}
\newtheorem{dfn}[thm]{Definition}

\theoremstyle{remark}
\newtheorem{rem}[thm]{Remark}

\begin{document}

\title{Uniqueness of steady states of Gorini-Kossakowski-Sudarshan-Lindblad equations:\\ a simple proof}

\author{Hironobu Yoshida}
\email{hironobu-yoshida57@g.ecc.u-tokyo.ac.jp}
\affiliation{Department of Physics, Graduate School of Science, The University of Tokyo, 7-3-1 Hongo, Tokyo 113-0033, Japan}
\begin{abstract}
We present a simple proof of a sufficient condition for the uniqueness of non-equilibrium steady states of Gorini-Kossakowski-Sudarshan-Lindblad equations. 
We demonstrate the applications of the sufficient condition using examples of the transverse-field Ising model, the XYZ model, and the tight-binding model with dephasing.
\end{abstract}

\maketitle

\section{Introduction}

Recent advances in quantum engineering have brought renewed interest in the effect of dissipation on quantum many-body systems. Under the Markov approximation, the dynamics of an open quantum system is described by the Gorini-Kossakowski-Sudarshan-Lindblad (GKSL) equation~\cite{gorini_completely_1976,lindblad_generators_1976,breuer_theory_2007}. 
Throughout this paper, we consider quantum systems described by a $d$-dimensional Hilbert space $\mathcal{H}$. By writing the set of linear operators on $\mathcal{H}$ as $\mathcal{B}(\mathcal{H})$, a state of the system is described by a density operator $\rho \in \mathcal{B}(\mathcal{H})$ that is Hermitian, positive semidefinite, and $\Tr \rho=1$. Then, the GKSL equation reads
\begin{equation}
  \frac{d \rho}{d \tau}=\mathcal{\hat{L}} (\rho)=-i[H, \rho]+\sum_{m=1}^M \left(L_{m} \rho L_{m}^{\dagger}-\frac{1}{2}\left\{L_{m}^{\dagger} L_{m}, \rho\right\}\right).
  \label{eq:Lindblad_matrix}
\end{equation}
Here, $\tau$ is the time, $H$ is the Hamiltonian, and $L_m$ $(m=1,\ldots,M)$ are the Lindblad operators that act on $\mathcal{H}$.

We write the eigenvalues of $\mathcal{\hat{L}}$ as $\Gamma_j$ and corresponding eigenmodes as $\rho_j$. Then, $\Re [\Gamma_j] \leq 0$ for all $\Gamma_j$.
A non-equilibrium steady state (NESS) $\rho_\infty$ is a density operator  that is an eigenoperator of $\mathcal{\hat{L}}$
with eigenvalue $0$. 
There always
exists at least one NESS in a finite-dimensional system. However, whether the NESS is unique or not depends on the system.
Frigerio~\cite{frigerio_quantum_1977,frigerio_stationary_1978} gave an algebraic criterion for the uniqueness of the NESS with the assumption that there exists a positive definite (or full-rank) NESS. See Appendix \ref{sec:frigerio} for the details of the result. See also related results by 
Spohn~\cite{spohn_approach_1976,spohn_algebraic_1977} and Evans~\cite{evans_irreducible_1977}, Ref.~\cite{spohn_kinetic_1980,nigro_uniqueness_2019} for a review of these works, Ref.~\cite{prosen_comments_2012} for an application, and Ref.~\cite{baumgartner_analysis_2008,baumgartner_analysis_2008-1,wolf_2010,buca_note_2012,burgarthErgodicMixingQuantum2013, albert_symmetries_2014,albert_geometry_2016,zhang_stationary_2020, amato_2023} for recent progress in understanding the degeneracy of the NESSs.

In this paper, we provide a proof of a sufficient condition for the uniqueness of NESS. Compared with Frigerio's theorem, our theorem does not require any prior information about the NESS. While a sufficient condition for general infinite-dimensional systems is presented in Ref.~\cite{fagnola_subharmonic_2002}, the proof provided there requires knowledge of von Neumann algebras. In contrast, our paper focuses solely on finite-dimensional systems. The significant advantage of such a limitation is that our proof for the sufficient condition is much more concise compared to Ref.~\cite{fagnola_subharmonic_2002}, and readers are only expected to possess an elementary knowledge of linear algebra to comprehend the proof.
Next, we see that the sufficient condition can also be used to study the steady-state degeneracy of systems with strong symmetries. In the presence of a strong symmetry, there is at least one NESS in every symmetry sector~\cite{buca_note_2012,zhang_stationary_2020}.
We give a sufficient condition for the uniqueness of the NESS in every symmetry sector.
Finally, we demonstrate the applications of the sufficient condition using examples of the transverse-field Ising model, the XYZ model, and the tight-binding model with dephasing.

\section{Main theorem}
\begin{thm}
\label{thm:main1}
If the set of operators $\left\{H-\frac{i}{2}\sum_{m=1}^M L^\dagger_m L_m, L_1, \ldots, L_M\right\}$ generates 
all the operators 
under multiplication, addition, and scalar multiplication, then $\rho_\infty$ is unique and positive definite.
\end{thm}
To prove Theorem \ref{thm:main1}, we prove the following lemma. 

\begin{lem}
\label{lem:main2}
Let $\rho$ be a positive semidefinite operator that satisfies $\mathcal{\hat{L}} (\rho)=0$. Under the same conditions as Theorem \ref{thm:main1}, $\rho$ is positive definite or zero.
\end{lem}

The following proofs of Theorem~\ref{thm:main1} and Lemma~\ref{lem:main2} are inspired by the method of spin reflection positivity~\cite{lieb_two_1989, tasaki_physics_2020}.

\begin{proof}[Proof of Lemma \ref{lem:main2}]
Assume that $\rho$ is positive semidefinite but not positive definite. Then, there exists a nonzero vector $\ket{\psi}\in \mathcal{H}$ such that $\rho\ket{\psi}=0$. By expanding $\bra{\psi} \mathcal{\hat{L}} (\rho) \ket{\psi}$, one finds
\begin{align}
    \bra{\psi} \mathcal{\hat{L}} (\rho) \ket{\psi}
    =\sum_{m=1}^M \bra{\psi} L_m \rho L^\dagger_m \ket{\psi}
    =\sum_{m=1}^M \|\sqrt{\rho} L^\dagger_m \ket{\psi}\|^2=0,
\end{align}
where $\sqrt{\rho}$ is a positive semidefinite operator such that $(\sqrt{\rho})^2=\rho$. Since $\|\sqrt{\rho} L^\dagger_m \ket{\psi}\|^2\geq0$, we have $\sqrt{\rho} L^\dagger_m \ket{\psi}=0$ for all $m$, which means that $\rho L^\dagger_m \ket{\psi}=0$ for all $m$. Next, by expanding $\mathcal{\hat{L}} (\rho) \ket{\psi}$, one obtains 
\begin{align}
    \mathcal{\hat{L}} (\rho)\ket{\psi}= i \rho \left(H+\frac{i}{2}\sum_{m=1}^M L^\dagger_m L_m\right)\ket{\psi}=0.
\end{align}
Therefore, if $\ket{\psi}\in \operatorname{Ker} \rho$, then $L^\dagger_m \ket{\psi}\in \operatorname{Ker} \rho$ for all $m$ and $\left(H+\frac{i}{2}\sum_{m=1}^M L^\dagger_m L_m\right)\ket{\psi}\in \operatorname{Ker} \rho$. By the assumption of Lemma \ref{lem:main2}, the set of operators $\left\{H+\frac{i}{2}\sum_{m=1}^M L^\dagger_m L_m,\ L^\dagger_1, \ldots,\ L^\dagger_M\right\}$ generates $\mathcal{B}(\mathcal{H})$, and therefore $\operatorname{Ker} \rho=\mathcal{H}$, which means that $\rho=0$.
\end{proof}

\begin{proof}[Proof of Theorem \ref{thm:main1}]
Assume that $\rho_1$ and $\rho_2$ ($\rho_1\neq \rho_2$) are NESSs. Since they are density operators, they are Hermitian, positive semidefinite, and $\Tr\rho_j=1$ $(j=1,2)$. Thus, by Lemma \ref{lem:main2}, they are positive definite. If we define~\footnote{We learned this method from Hosho Katsura.} 
\begin{equation}
    \rho_{\mathrm{un}}(x)=(1-x)\rho_1-x\rho_2 \quad (0\leq x\leq 1),
\end{equation}
$\rho_{\mathrm{un}}(0)=\rho_1$, $\rho_{\mathrm{un}}(1)=-\rho_2$, and $\mathcal{\hat{L}} (\rho_{\mathrm{un}}(x))={(1-x)\mathcal{\hat{L}}(\rho_1)}-{x\mathcal{\hat{L}}(\rho_2)}=0$ for all $x$ because $\mathcal{\hat{L}}(\rho_1)=\mathcal{\hat{L}}(\rho_2)=0$ by definition. Since all the eigenvalues of $\rho_{\mathrm{un}}(0)$ $[\rho_{\mathrm{un}}(1)]$ are positive [negative] and the spectrum of $\rho_{\mathrm{un}}(x)$ is continuous with respect to $x$, there exists a real number $0\leq x_0\leq 1$ such that the minimum eigenvalue of $\rho_{\mathrm{un}}(x_0)$ is zero. Namely, $\rho_{\mathrm{un}}(x_0)$ is positive semidefinite but not positive definite. Thus by Lemma \ref{lem:main2}, $\rho_{\mathrm{un}}(x_0)=0$. Then $\Tr \rho_{\mathrm{un}}(x_0)=1-2x_0=0$ and therefore $x_0=1/2$. Thus $\rho_{\mathrm{un}}(x_0)=(\rho_1-\rho_2)/2=0$. However, this contradicts the assumption $\rho_1\neq \rho_2$, so the NESS has to be unique.
\end{proof}

When all the Lindblad operators are Hermitian, the completely mixed state $\mathbb{I}_d/d$ is a NESS~\footnote{More generally, if there exists an $M\times M$ unitary matrix $u$ such that $L_m^\dagger = \sum_{n=1}^M u_{mn} L_n$, the completely mixed state $\mathbb{I}_d/d$ is a NESS}, where $\mathbb{I}_d$ is the identity matrix of size $d$. In this case, Theorem~\ref{thm:main1} boils down to the following corollary.

\begin{cor}
\label{cor:main3}
If all $L_m$ are Hermitian and the set of operators $\{H, L_1,\ldots,L_M\}$ generates 
all the operators
under multiplication, addition, and scalar multiplication, then $\rho_\infty$ is unique and $\rho_\infty=\mathbb{I}_d/d$.
\end{cor}

\section{Strong Symmetry}
Next, we consider systems with the strong symmetry~\cite{buca_note_2012,albert_symmetries_2014,zhang_stationary_2020}. 

\begin{dfn}[Strong symmetry]
The GKSL equation has a strong symmetry if there exists a unitary operator $S$ on $\mathcal{H}$ such that
\begin{equation}
    [S,H]=0 \text{ and } [S,L_m]=0 \text{ for all } m.
    \label{eq:strong_symmetry}
\end{equation}
\end{dfn}
We write $n_S$ different eigenvalues of $S$ as
$s_\alpha=e^{i \theta_\alpha}$ $(\alpha=1,\ldots, n_S)$, and corresponding eigenspace as 
$\mathcal{H}_\alpha$. 
Then, the following theorem is proved in Ref.~\cite{buca_note_2012}.
\begin{thm}[Bu\v{c}a and Prosen]
If there is a unitary operator $S$ that satisfies Eq. \eqref{eq:strong_symmetry}, then we obtain the following 1. and 2.:
\begin{enumerate}
    \item The space of operators $\mathcal{B}(\mathcal{H})$ can be decomposed into $n_S^2$ invariant subspaces of $\mathcal{\hat{L}}$:
    \begin{equation}
        \mathcal{\hat{L}}(\mathcal{B}_{\alpha,\beta})\subseteq \mathcal{B}_{\alpha,\beta},\  \mathcal{B}_{\alpha,\beta}=\left\{ \ket{\psi} \bra{\phi}; \ket{\psi} \in \mathcal{H}_\alpha, \ket{\phi} \in \mathcal{H}_\beta \right\} 
    \end{equation}
    for $\alpha,\beta=1,\ldots,n_S.$
    \item Every $\mathcal{B}_{\alpha,\alpha}$ contains at least one NESS:
    \begin{equation}
        \rho_\infty^\alpha \in \mathcal{B}_{\alpha,\alpha} \text{ for } \alpha=1,\ldots, n_S.
    \end{equation}
\end{enumerate}

\end{thm}

This theorem states that in the presence of strong symmetry, NESSs are always degenerate. However, we can apply Theorem \ref{thm:main1} to prove the uniqueness of the NESS in $\mathcal{B}_{\alpha,\alpha}$. When $H$ and $L_m$ commute with $S$, they can be decomposed as $H=\oplus_{\alpha=1}^{n_S} H|_{\mathcal{H}_\alpha}$ and $L_m=\oplus_{\alpha=1}^{n_S} L_m|_{\mathcal{H}_\alpha}$, where $H|_{\mathcal{H}_\alpha}$ and $L_m|_{\mathcal{H}_\alpha}$ are elements of $\mathcal{B}_{\alpha,\alpha}$. Then, if ${\{H-\frac{i}{2}\sum_{m=1}^M L^\dagger_m L_m, L_1, \ldots, L_M\}}$ generates all the operators that commute with $S$, the set of operators $\{H-\frac{i}{2}\sum_{m=1}^M L^\dagger_m L_m|_{\mathcal{H}_\alpha}, L_1|_{\mathcal{H}_\alpha}, \ldots, L_M|_{\mathcal{H}_\alpha}\}$ generates $\mathcal{B}_{\alpha,\alpha}$ for all $\alpha$~\footnote{Note that an operator $A\in \mathcal{B}(\mathcal{H})$ commutes with $S$ if and only if $A$ is expressed as $A=\oplus_{\alpha=1}^{n_S}A_\alpha$, where $A_\alpha \in \mathcal{B}_{\alpha,\alpha}$.}. By applying Theorem \ref{thm:main1} to $\mathcal{H}_\alpha$ and writing $\operatorname{dim} \mathcal{H}_\alpha=d_\alpha$, 
we have the following corollaries:

\begin{cor}
\label{cor:sym_1}
    If the set of operators $\left\{H-\frac{i}{2}\sum_{m}L^\dagger_m L_m, L_1,\ldots, L_M\right\}$ generates all the operators that commute with $S$ under multiplication, addition, and scalar multiplication, then $\rho_\infty^\alpha|_{\mathcal{H}_\alpha}$ is unique and positive definite for all $\alpha$.
\end{cor}

\begin{cor}
\label{cor:sym_2}
    If all $L_m$ are Hermitian and the set of operators $\{H, L_1, \ldots, L_M\}$ generates all the operators that commute with $S$ under multiplication, addition, and scalar multiplication, then $\rho_\infty^\alpha|_{\mathcal{H}_\alpha}$ is unique and $\rho_\infty^\alpha|_{\mathcal{H}_\alpha}=\mathbb{I}_{d_\alpha}/{d_\alpha}$ for all $\alpha$.
\end{cor}

\section{Examples}

In this section, we demonstrate the applications of Theorem \ref{thm:main1}, Corollaries \ref{cor:main3} and \ref{cor:sym_2}. As the simplest example, we consider the two-level system with gain and loss. Next, we present an application of Corollary \ref{cor:main3} to the transverse-field Ising model with boundary dephasing. Finally, we present applications of Corollary \ref{cor:sym_2} to the XYZ model and the tight-binding model with bulk dephasing, as prototypical examples of models with $\mathbb{Z}_2$ and U(1) strong symmetries. We also note that Theorem~\ref{thm:main1} can be applied to the boundary-driven open XXZ chain~\cite{znidaric_dephasing-induced_2010, prosen_open_2011, prosen_comments_2012,
znidaric_diffusive_2016}.

\subsection{Two-level system with gain and loss}
As our first example, we consider a two-level system with gain and loss. We write an orthonormal basis of $\mathcal{H}=\mathbb{C}^2$ as $\ket{\uparrow}$ and $\ket{\downarrow}$. The Lindblad operators of gain and loss are $L_g=\sqrt{\gamma_g}\ket{\uparrow}\bra{\downarrow}$ and $L_l=\sqrt{\gamma_l}\ket{\downarrow}\bra{\uparrow}$. Then, $L_g L_l\propto \ket{\uparrow}\bra{\uparrow}$ and $L_l L_g\propto \ket{\downarrow}\bra{\downarrow}$, and therefore $L_g$, $L_l$, $L_gL_l$ and $L_lL_g$ form the basis of $\mathcal{B}(\mathbb{C}^2)$. From Theorem \ref{thm:main1}, the NESS $\rho_\infty$ is unique and positive definite for an arbitrary Hamiltonian. 

\subsection{Transverse-field Ising model with boundary dephasing}\label{sec:TFIM}
Next, we consider the spin-$1/2$ transverse-field Ising chain under open boundary conditions~\cite{vasiloiu_enhancing_2018}
\begin{equation}
    H=\sum_{j=1}^{N-1} \sigma_{j}^{z} \sigma_{j+1}^{z}+h_x \sum_{j=1}^{N}\sigma^x_j
\end{equation}
with dephasing noise $L_1=\sqrt{\gamma}\sigma^z_1$ at the first site of the lattice. Here, $\sigma_j^\alpha$ $(\alpha=x,y,z)$ are the Pauli operators at site $j=1,\ldots,N$ acting on $d=2^N$ dimensional Hilbert space $\mathcal{H}$, $h_x\neq 0$ is the external magnetic field, $\gamma>0$ is the dissipation strength parameter.

By using Corollary \ref{cor:main3}, we can prove that the NESS $\rho_\infty$ is unique and written as $\rho_\infty=\mathbb{I}_{2^N}/{2^N}$.
\begin{proof}
We first note that $L_1\propto \sigma_1^z$, $[\sigma_1^z, H]\propto \sigma_{1}^y$, and then $\frac{1}{2}\sigma_{1}^y[\sigma_{1}^y\sigma_{1}^z,H]=\sigma_{2}^z$, $[\sigma_{2}^z,H]\propto \sigma_{2}^y$. Next, we observe the recurrence relation
\begin{align}
    \frac{1}{2}\sigma_{j}^y[\sigma_{j}^y\sigma_{j}^z, H]- \sigma_{j-1}^z &= \sigma_{j+1}^z, \\
    [\sigma_{j+1}^z,H]&\propto \sigma_{j+1}^y
\end{align}
for $j=2,\ldots, N-1$, which generates $\sigma_{j}^y$ and $\sigma_{j}^z$ for all $j$ from $L_1$ and $H$. Then, it is clear that they generate all the operators in $\mathcal{B}(\mathcal{H})$, and therefore the NESS is unique and written as $\rho_\infty=\mathbb{I}_{2^N}/{2^N}$ from Corollary \ref{cor:main3}.
\end{proof}

\subsection{XYZ model with bulk dephasing}\label{sec:XYZ}
As a prototypical example of models with the $\mathbb{Z}_2$ strong symmetry, we consider the spin-$1/2$ XYZ chain under periodic boundary conditions~\cite{vasiloiu_enhancing_2018,wellnitz_rise_2022}
\begin{equation}
    H=\sum_{j=1}^{N} (J_{x} \sigma_{j}^{x} \sigma_{j+1}^{x}+J_{y} \sigma_{j}^{y} \sigma_{j+1}^{y}+J_{z} \sigma_{j}^{z} \sigma_{j+1}^{z})+h_z \sum_{j=1}^{N}\sigma^z_j
\end{equation}
with dephasing strength $L_j=\sqrt{\gamma}\sigma^z_j$ at every site $j$. 
Here, $J_\alpha \in \mathbb{R}$ $(\alpha=x,y,z)$ are the exchange couplings, $h_z\in \mathbb{R}$ is the external magnetic field, $\gamma>0$ is the dissipation strength parameter, and $N$ is the number of sites. 
We assume that $|J_x|\neq |J_y|$~\footnote{If $J_x=-J_y$ with even $N$ or $J_x=J_y$, the system has the U(1) strong symmetry~\cite{znidaric_relaxation_2015}. Then we can prove that the NESS is unique in every $\mathcal{B}_{\alpha,\alpha}$ in the same manner as in Sec.~\ref{sec:tb}}.
By defining a unitary operator $S=\prod_{j=1}^N \sigma_{j}^{z}$ , one finds
\begin{equation}
    [S,H]=0 \text{ and } \left[S, L_{j}\right]=0 \text { for all } j.
\end{equation}
The eigenvalues of $S$ are $\pm 1$, and we write the corresponding subspace of operators as $\mathcal{B}_{\alpha,\beta}\ (\alpha,\beta=\pm)$.
If we define $\rho^{ \pm}$ as $\rho^{ \pm} \coloneqq\left(\mathbb{I}_{2^N} \pm S\right) / 2^{N}$, it can be checked that $\rho^{+} \in \mathcal{B}_{+,+}$, $\rho^{-} \in \mathcal{B}_{-,-}$, and $\hat{\mathcal{L}}(\rho^{ \pm})=0$. By using Corollary \ref{cor:sym_2}, we prove that they are the unique NESSs in $\mathcal{B}_{+,+}$ and $\mathcal{B}_{-,-}$, respectively.

\begin{proof}
First, we identify all the operators that commute with $S$. From the relations
\begin{equation}
    S \sigma_{j}^{x}=-\sigma_{j}^{x}S,\  S \sigma_{j}^{y}=-\sigma_{j}^{y}S,\  S \sigma_{j}^{z}=\sigma_{j}^{z}S,
\end{equation}
all the operators that commute with $S$ are spanned by products of an even number of $\sigma_{j}^{x}$ and $\sigma_{j}^{y}$.
Thus it is sufficient to prove that $\sigma^{\mu}_j \sigma^{\nu}_k\ (\mu,\nu=x,y)$ can be generated by $H$ and $L_j$ for all $1\leq j\leq k\leq N$. When $j=k$, it can be generated only by $L_j$, because $\sigma^{x}_j \sigma^{y}_j=-\sigma^{y}_j \sigma^x_j\propto L_j$ and $(\sigma^{x}_j)^2=(\sigma^{y}_j)^2\propto (L_j)^2$. Next, we consider the cases where $j\neq k$. First, one finds 
\begin{align}
    A_1&\coloneqq[\sigma^{z}_l,H]\propto \sum_{\sigma=\pm1} (J_x \sigma_l^y \sigma_{l+\sigma}^x-J_y \sigma_l^x \sigma_{l+\sigma}^y), \\
    A_2&\coloneqq[\sigma^{z}_{l+1},A_1]\propto  (J_x \sigma_l^y \sigma_{l+1}^y+J_y \sigma_l^x \sigma_{l+1}^x), \\
    A_3&\coloneqq[\sigma^{z}_l,A_2]\propto  (J_x \sigma_l^x \sigma_{l+1}^y-J_y \sigma_l^y \sigma_{l+1}^x), \\
    A_4&\coloneqq[\sigma^{z}_{l+1},A_3]\propto  (J_x \sigma_l^x \sigma_{l+1}^x+J_y \sigma_l^y \sigma_{l+1}^y), \\
    A_5&\coloneqq[\sigma^{z}_l,A_4]\propto  (J_x \sigma_l^y \sigma_{l+1}^x-J_y \sigma_l^x \sigma_{l+1}^y).
\end{align}
Noting that $|J_x|\neq |J_y|$, we obtain $\sigma_l^x \sigma_{l+1}^x$ and $\sigma_l^y \sigma_{l+1}^y$ by linear combinations of $A_2$ and $A_4$ and $\sigma_l^x \sigma_{l+1}^y$ and $\sigma_l^y \sigma_{l+1}^x$ by linear combinations of $A_3$ and $A_5$. Finally, since 
\begin{equation}
    \sigma^{\mu}_j \sigma^{\nu}_k= \sigma^{\mu}_j  \sigma^{x}_{j+1}\left(\prod_{l=j+1}^{k-2} 
    \sigma^{x}_l  \sigma^{x}_{l+1} \right)\sigma^{x}_{k-1}
    \sigma^{\nu}_{k},
\end{equation}
we have $\sigma^{\mu}_j \sigma^{\nu}_k\ (\mu,\nu=x,y)$ for all $1\leq j\leq k\leq N$, and thus we have all possible products of an even number of $\sigma_{j}^{x}$ and $\sigma_{j}^{y}$ $(j=1,2,\ldots, N)$. Therefore, from Corollary \ref{cor:sym_2}, the NESS is unique in $\mathcal{B}_{+,+}$ and $\mathcal{B}_{-,-}$, respectively.
\end{proof}

\begin{rem}
For simplicity, we assumed that the Hamiltonian is one-dimensional and translationally invariant. However, these assumptions are not necessary. To illustrate this, we write the set of sites as $\sites$ and the set of bonds as $\bonds$, and consider the following Hamiltonian on a general lattice $(\sites,\bonds)$:
\begin{equation}
    H=\sum_{j,k \in \sites} (J^{x}_{j,k} \sigma_{j}^{x} \sigma_{k}^{x}+J^{y}_{j,k}\sigma_{j}^{y} \sigma_{k}^{y}+J^{z}_{j,k} \sigma_{j}^{z} \sigma_{k}^{z})+\sum_{j\in \sites}h^z_j \sigma^z_j
    \label{eq:xyz_general}
\end{equation}
with dephasing noise $L_j=\sqrt{\gamma_j}\sigma^z_j$, where $\gamma_j>0$ for all $j\in \Lambda$. We assume that $|J^x_{j,k}| \neq |J^y_{j,k}|$ when $(j,k)\in \bonds$ and $J^x_{j,k}=J^y_{j,k}=0$ when $(j,k)\notin \bonds$. If the lattice $(\sites,\bonds)$ is connected~\footnote{
A lattice $(\sites,\bonds)$ is connected if for any $j,k \in \sites$ such that $j\neq k$, there exists a finite sequence $l_1,\ldots,l_n \in \sites$ such that $l_1=j$, $l_n=k$, and $(l_i,l_{i+1})\in \bonds$ for $i = 1,\ldots,n-1$.}, one can prove that the NESS is unique in $\mathcal{B}_{+,+}$ and $\mathcal{B}_{-,-}$, respectively. For example, equation \eqref{eq:xyz_general} includes the one-dimensional quantum compass model
\begin{equation}
   H=-\sum_{j=1}^{N / 2} J_{x} \sigma_{2 j-1}^{x} \sigma_{2 j}^{x}-\sum_{j=1}^{N / 2-1} J_{y} \sigma_{2 j}^{y} \sigma_{2 j+1}^{y} 
\end{equation}
with dephasing noise $L_j=\sqrt{\gamma}\sigma^z_j$ discussed in Ref. \cite{shibata_dissipative_2019}.
\end{rem}

\subsection{Tight-binding model with bulk dephasing}\label{sec:tb}

Finally, we consider the tight-binding chain under the periodic boundary conditions~\cite{medvedyeva_exact_2016}
\begin{equation}
  H= t \sum_{j=1}^{N} (c^\dagger_j c_{j+1}+c^\dagger_{j+1} c_j)+\delta \mathbb{I}_{2^N}
  \label{eq:spinless_dephasing}
\end{equation}
with dephasing noise $L_j=\sqrt{\gamma}n_j$ at every site $j$. Here, $c^\dagger_j$ and $c_j$ are the creation and annihilation operators, respectively, of a fermion at site $j=1,\ldots, N$, $n_j=c^\dagger_j c_j$ is the number operator, $t\neq 0$ is the hopping amplitude, $\gamma>0$ is the dephasing strength, and $\delta$ is a real constant. The eigenvalues and eigenmodes of $\hat{\mathcal{L}}$ do not depend on the constant $\delta$, but we assume that $\delta\neq 0$ to simplify the proof. If we write the vacuum state annihilated by all $c_j$ as $\ket{0}$, then the Hilbert space $\mathcal{H}$ is spanned by states of the form $\{ \prod_{j=1}^L (c^\dagger_j)^{m_j} \}\ket{0}$ ${({m_j}=0,1)}$.

Next, we write the total number operator as $N_\mathrm{tot}=\sum_{j =1}^N n_j$ and define a unitary operator $S=e^{i N_\mathrm{tot}}$. Then, the eigenvalues of $S$ are $e^{i \alpha}$ $(\alpha=0,1,\ldots, N)$. Since $S$ commutes with $H$ and all $L_j$, $\mathcal{B}(\mathcal{H})$ can be decomposed into invariant subspaces of $\mathcal{\hat{L}}$ and we write them as $\mathcal{B}_{\alpha,\beta}$ $(\alpha,\beta=0,1,\ldots, N)$. Then, we prove that the NESS is unique in every $\mathcal{B}_{\alpha,\alpha}$.

\begin{proof}
Since all the operators that commute with $S$ are written as a sum of monomials that are products of the same number of creation and annihilation operators, it is sufficient to prove that $\mathbb{I}_{2^N}$ and $c^\dagger_j c_k$ can be generated by $H$ and $L_j$ for all $1\leq j,k\leq N$. When $j=k$, $c^\dagger_j c_j$ is proportional to $L_j$, so we concentrate on the case $j\neq k$. Without loss of generality, we can assume that $j<k$. First, we see that the following commutation relations hold:
\begin{align}
[n_l,H]&= t \sum_{\sigma=\pm1}(c^\dagger_l c_{l+\sigma}-c^\dagger_{l+\sigma} c_l), \\
[n_{l+1},[n_l,H]]&= -t(c^\dagger_l c_{l+1}+c^\dagger_{l+1} c_{l}), \\
[n_{l},[n_{l+1},[n_l,H]]]&= -t (c^\dagger_l c_{l+1}-c^\dagger_{l+1} c_{l}).
\end{align}
Therefore, we obtain $c^\dagger_l c_{l+1}$ and $c^\dagger_{l+1} c_l$ from multiplication, addition, and scalar multiplication of $H$, $L_l$, and $L_{l+1}$. 
Since $[c^\dagger_l c_m,c^\dagger_m c_n]=c^\dagger_l c_n$ when $l\neq n$, we can generate $c^\dagger_j c_k$ for any $1\leq j<k\leq N$ with $c^\dagger_j c_{j+1}$, \ldots, $c^\dagger_{k-1} c_{k}$.
Finally, $\mathbb{I}_{2^N}$ can be obtained by $[H-t \sum_{j=1}^{N} (c^\dagger_j c_{j+1}+c^\dagger_{j+1} c_j)]/\delta$. Therefore, from Corollary \ref{cor:sym_2}, the NESS is unique in every $\mathcal{B}_{\alpha,\alpha}$.
\end{proof}

\begin{rem}
The result can be generalized to the tight-binding model on a general lattice $(\sites,\bonds)$ with $N$ sites:
\begin{equation}
  H= \sum_{j, k \in \sites} t_{j,k} c^\dagger_j c_k+\delta\mathbb{I}_{2^N} 
  ,\quad L_j= \sqrt{\gamma_j}n_j,
\end{equation}
where $H$ is Hermitian, i.e., $t_{j,k}=t^*_{k,j}$, $\gamma_j>0$ for all $j\in \Lambda$, and $\delta$ is a real constant.
We also assume that $t_{j,k}\neq 0$ when $(j,k)\in \bonds$ and $t_{j,k}=0$ when $(j,k)\notin \bonds$. When the lattice $(\sites,\bonds)$ is connected, one can prove that the NESS is unique in every $\mathcal{B}_{\alpha,\alpha}$.
\end{rem}

\vspace{5mm}

\section{Conclusion}
We presented a simple proof of a sufficient condition for the uniqueness of NESSs of GKSL equations. We also presented applications of the sufficient condition to the transverse-field Ising model, the XYZ model, and the tight-binding model with dephasing. 
Our results here open many interesting questions. The most important direction for future study is to generalize our proof to the sufficient and necessary condition for the uniqueness of the NESS.
Another direction is to apply the sufficient condition to clarify the degeneracy of the NESS in the presence of the non-abelian strong symmetries~\cite{zhang_stationary_2020} or the hidden strong symmetries in the form of quasi-local charges~\cite{de_leeuw_hidden_2023}.

\begin{acknowledgments}
We acknowledge fruitful comments by Franco Fagnola, Wen Wei Ho, Hosho Katsura, Zala Lenar\v{c}i\v{c}, Leonardo Mazza, Ken Mochizuki, Tadahiro Miyao, 
Toma\v{z} Prosen, Herbert Spohn,  Gergely Zar\'{a}nd, and Marko \v{Z}nidari\v{c}. H.Y. was supported by JSPS KAKENHI Grant-in-Aid for JSPS fellows Grant No. JP22J20888, the Forefront Physics and Mathematics Program to Drive Transformation, and JSR Fellowship, the University of Tokyo.
\end{acknowledgments}

\appendix

\section{Frigerio's theorem}
\label{sec:frigerio}
In this section, we briefly review Frigerio's theorem on the uniqueness of the NESS. While his result applies to general infinite-dimensional systems, here we state the theorem in the $d$-dimensional case. For a set of operators $\mathcal{A} \subseteq \mathcal{B}(\mathcal{H})$, let us denote by $\mathcal{A}^\prime$ the commutant of the set $\mathcal{A}$, i.e., the set of operators that commute with all the elements of $\mathcal{A}$. Frigerio~\cite{frigerio_quantum_1977,frigerio_stationary_1978} proved that if there exists a positive definite NESS $\rho_\infty$, then $\rho_\infty$ is the unique NESS iff $\left\{H, L_1,\ldots, L_M, L^\dagger_1,\ldots, L^\dagger_M\right\}^\prime=\left\{c \mathbb{I}_d | c\in \mathbb{C}\right\}$, where $\mathbb{I}_d$ is the identity matrix of size $d$. 
This condition is equivalent to the condition that the set of operators $\left\{H, L_1,\ldots, L_M, L^\dagger_1,\ldots, L^\dagger_M\right\}$ generates all the operators
under multiplication, addition, and scalar multiplication. 
Note that the assumption of the existence of a positive definite NESS is necessary, and without this assumption, several counterexamples can be found~\cite{zhang_2023}.

\bibliographystyle{apsrev4-1}
\bibliography{reference2}

\end{document}